\documentclass[pra,aps,epsf,nofootinbib,
notitlepage,showpacs,10pt]{revtex4-1}

\usepackage{color}
\usepackage{float}
\usepackage{graphicx}
\usepackage{amsmath}
\usepackage{amsthm}
\usepackage{mathtools}
\usepackage{algorithm}
\usepackage{algorithmic}
\usepackage[normalem]{ulem}
\newtheorem{theorem}{Theorem}

\newtheorem{example}{Example}
\theoremstyle{definition}
\usepackage{subcaption}
\usepackage{bm}

\newcommand{\ket}[1]{{\left| #1 \right\rangle}}

\theoremstyle{definition}
\usepackage{tikz} 
\usetikzlibrary{matrix,positioning}
\usepackage{amsfonts}
\usepackage{amssymb}
\begin{document}

\markboth{Rebekah Herrman}
{Globally optimizing QAOA circuit depth for constrained optimization problems}

\title{Globally optimizing QAOA circuit depth for constrained optimization problems}
	
\author{Rebekah Herrman}\thanks{Corresponding author}
\email{rherrma2@tennessee.edu}
\affiliation{
	Department of Industrial and Systems Engineering, University of Tennessee at Knoxville\\Knoxville, Tennessee  37996 USA}
	
\author{Lorna Treffert}
\email{ltreffer@vols.utk.edu}
\affiliation{
	Department of Industrial and Systems Engineering, University of Tennessee at Knoxville\\Knoxville,  Tennessee 37996 USA}

\author{James Ostrowski}
\email{jostrows@tennessee.edu}
\affiliation{
	Department of Industrial and Systems Engineering, University of Tennessee at Knoxville\\Knoxville,  Tennessee  37996 USA}

\author{Phillip C. Lotshaw}
\email{lotshawpc@ornl.gov}
\affiliation{
	Quantum Computing Institute\\ Oak Ridge National Laboratory\\ Oak Ridge,  Tennessee 37830 USA}
	
\author{Travis S. Humble}
\email{humblets@ornl.gov}
\affiliation{
	Quantum Computing Institute\\ Oak Ridge National Laboratory\\ Oak Ridge,  Tennessee 37830 USA}

\author{George Siopsis}
\email{siopsis@tennessee.edu}
\affiliation{
	Department of Physics and Astronomy, University of Tennessee at Knoxville\\Knoxville, Tennessee 37996-1200 USA}

\begin{abstract}

We develop a global variable substitution method that reduces $n$-variable monomials in combinatorial optimization problems to equivalent instances with monomials in fewer variables. We apply this technique to $3$-SAT and analyze the optimal quantum circuit depth needed to solve the reduced problem using the quantum approximate optimization algorithm.
For benchmark $3$-SAT problems, we find that the upper bound of the circuit depth is smaller when the problem is formulated as a product and uses the substitution method to decompose gates than when the problem is written in the linear formulation, which requires no decomposition. 
\end{abstract}

\maketitle
\section{Introduction}

The quantum approximate optimization algorithm (QAOA) was introduced to approximately solve combinatorial optimization problems \cite{farhi2014quantum, farhi2014bounded}. QAOA research has mostly focused on a small subset of combinatorial optimization (CO) problems such as MaxCut, MaxIndSet, and Max k-cover \cite{lotshaw2021bfgs, herrman2021impact, saleem2020, wang2018quantum, crooks2018performance, guerreschi2019qaoa, cook2020quantum}.
These problems can be easily written as quadratic unconstrained binary optimization (QUBO) problems by identifying each variable with a qubit.
QUBOs are implementable on current hardware \cite{ryan2017hardware, linke2017experimental}. Recent work has examined how QAOA can be used on CO problems that can be written as polynomial unconstrained binary optimization problems \cite{liu2021layer, guerreschi2021solving}. When solving a CO problem using QAOA, each monomial in $k$ vertices in the problem formulation corresponds to a $k$ qubit gate. The circuit depth for one layer of QAOA for combinatorial optimization (PUBO) problems 
was shown in \cite{herrman2021lower} to be the edge chromatic number of the graph, or hypergraph, derived from these problems. When the combinatorial optimization problem can be written as a QUBO, the derived graph is not a hypergraph, so the edge chromatic number is either the maximum degree, or the maximum degree plus one \cite{vizing1964estimate}. When the derived graph is a hypergraph, the edge chromatic number may be more difficult to compute.


Some classes of combinatorial optimization problems, however, can be written in more than one way, where each formulation may have monomials of different sizes. For example, Boolean satsfiability (SAT) problems can be written as a QUBO or a more general PUBO. SAT problems have been studied extensively in a classical setting \cite{tovey1984simplified, marques2008practical} and form the backbone of complexity theory. They can be written in a linear form, which requires two-qubit gates to implement, or a product form, which requires larger multi-qubit gates. Since these larger multi-qubit gates are not easily implementable on current hardware, polynomial formulations must be decomposed into sums of two variable monomials. We show in the context of 3-SAT that decomposing the PUBO can lead to a shallower circuit needing fewer qubits than using the natural QUBO formulation~\cite{herrman2021lower}. This has implications for more general combinatorial problems, as modeling approaches can have a significant impact on the design of the resulting circuit.

In this paper, we first review QAOA, the classical linear and product formulations of SAT problems, and how they translate to QUBOs in Sec.~\ref{review}. Then, in Sec.~\ref{variablesub}, we introduce a substitution method, called the \textit{global variable substitution} (GVS) method, to decompose monomials consisting of $k \geq 3$ variables into ones that can be implemented on current hardware. 
Next, we discuss how to optimize GVS for $3$-SAT problems in Sec.~\ref{optimization} and then apply this work to instances  from the SATLIB Benchmark Problems suite~\cite{hoos2000satlib} in Sec.~\ref{computational}. Finally, we summarize the results and discuss future work in Sec.~\ref{conclusion}. 

\section{QAOA and Combinatorial Optimization Problem Review}\label{review}

In this section, we review QAOA, dualization, and the 3-SAT problem. 

\subsection{QAOA}

In order to use QAOA to solve a CO problem, we apply two operators, $U(C, \gamma) = e^{-iC\gamma}$ and $U(B, \beta) = e^{-iB\beta}$, in succession on an initial state. The initial state is the uniform superposition, $\ket{s}=\frac{1}{\sqrt{2^n}} \sum_{z} \ket{z}$, where the sum is over the computation basis $\ket{z}$. 
The outcome of one iteration of QAOA is
\begin{equation*}
\ket{\gamma,\beta}=U(B, \beta)U(C,\gamma)\ket{s}.
\end{equation*}
Here $C$ encodes the problem to be solved and $B$ is a mixing operator. Often, $C$ is the sum over a collection of clauses, $C = \sum_{a} C_a$, and $B$ is typically $B = \sum_{v \in V(G)} B_v$, where $B_v = \sigma_v^x$ is the Pauli-X operator acting on the $v^{th}$ qubit. For more detail about QAOA, see \cite{farhi2014quantum, farhi2014bounded, farhisupremacy}.

There is a direct correlation between the circuit depth and level of noise in a quantum circuit \cite{xue2019effects, wang2020noise, marshall2020characterizing}, so it is important to consider the circuit depth when developing a circuit that implements QAOA.  In the next subsection, we review dualization and the related previous work that determines the circuit depth of one iteration of QAOA.

\subsection{Dualization and Circuit Depth}
Previously \cite{herrman2021lower}, we considered the method of dualizing constraints to solve CO problems of the form
\begin{align}
    & \min \  c(x) \label{eq:obj}\\
    & \mbox{s.t. }  p_i (x) \leq b_i & \forall i \in P\label{eq:constraint}\\
    & x \in \{0,1\}^n\label{eq:constrainttwo}
\end{align}

\noindent where $p_i$ is contained in the collection of polynomial constraints $P$. Both $p_i$ and $c$ are polynomial functions in $\mathbb{R}^n[x_1, x_2, ... , x_n ]$ and $b_i \in \mathbb{R}$. We eliminate constraints via dualization by subtracting $b_i$ from both sides, adding in slack variables, squaring both sides, and adding the new expression to the objective function \cite{herrman2021lower}. The \textit{derived (hyper)graph} from dualization is one in which there is a vertex for each variable and (hyper)edges between variables that appear in a monomial together. For example, if $x_1x_2$ appears in the dualization, there is an edge between vertices $x_1$ and $x_2$. 
 
In graph theory, a proper edge coloring of a graph is a function $f:E(G) \rightarrow [k]$ such that if two edges share a vertex, they receive different values, e.g. $uv$ and $xv$, $f(uv) \neq f(xv)$. Given the set $[k] = \{1, ... , k\}$, each element of this set can be thought of as a color, hence the name edge coloring. The smallest number $m$ such that a proper edge coloring of a graph $G$ is possible with $m$ colors is denoted $\chi'(G)$ and called either the chromatic index or the edge chromatic number. It has been shown that chromatic index of the derived (hyper)graph is directly related to the circuit depth of one iteration of QAOA, if the size of each monomial in the objective function is at most the gate size the hardware can support \cite{herrman2021lower}. In particular, the following theorem holds:

\begin{theorem}[\cite{herrman2021lower}]
Let $H$ be the hypergraph derived from a combinatorial optimization problem instance. Every proper edge coloring of $H$ corresponds to a valid circuit for a $\textrm{PUBO}$, where the depth of the shallowest circuit is $\chi'(H)+1$.
\end{theorem}

\subsection{SAT}

A Boolean satisfiability problem (SAT) is one type of CO problem that has the form of Eqs.~\eqref{eq:obj}-\eqref{eq:constrainttwo}. It is defined by a collection of clauses, $C$, consisting of $N$ literals. The goal is to determine if the values TRUE or FALSE can be assigned to each literal in a clause such that every clause is TRUE. This problem is classically NP-complete \cite{karp1972reducibility}, even when each clause contains only three literals. When each clause contains precisely three literals, the problem is known as a $3$-SAT problem. We will use $3$-SAT as the key example throughout the paper. Each clause has the form $y_i \vee y_j \vee y_k$ where $y_i \in \{x_i, \backsim x_i\}$. 
Now, we look at two formulations of $3$-SAT problems.

\subsubsection{SAT linear formulation}\label{linear}

Let $\{z_c\}_{c \in C}$ be the collection of variables indicating if clause $c$ is satisfied, and $x_i$ be the indicator variable denoting if literal $i$ is satisfied
.  Let $TRUE_c$ be the set of literals that must be true to satisfy $c$, and $FALSE_c$ the rest. Then, $3$-SAT can be written as
\begin{align}
    & \max\  \sum_{c \in C} z_c & \label{satform1}\\
    & \mbox{s.t. } \sum_{x_{i}= TRUE_c } x_{i} +\sum_{x_{i}= FALSE_c} (1-x_{i}) \geq z_c & \forall c \in C \label{satform2}\\
    &x_i,\ z_c \in \{0,1\} \label{satform3} .
\end{align}

Equation~\eqref{satform2} has this form because $3$-SAT is defined to have only OR conditions, so at least one $x_i$ must be satisfied when $z_c=1$. Standard QAOA does not have a way to incorporate constraints; we therefore adopt  the aforementioned strategy of dualizing the constraints to yield an unconstrained optimization problem suitable for QAOA as follows.

Before we can dualize Eq.~\eqref{satform2}, we must change the inequality so it matches that of Eq.~\eqref{eq:constraint}. We therefore take the contrapositive of each constraint so Eq.~\eqref{satform2} has the form $\sum_{x_{i} = TRUE} (1- x_{i}) + \sum_{x_{i} = FALSE} x_{i} \leq 2 +z_c$.  We can express this constraint as an equality, 
\begin{equation*}
    \sum_{x_{i} = TRUE} (1- x_{i}) + \sum_{x_{i}= FALSE} x_{i} + \delta_{c,1} + \delta_{c,2} = 2 +z_c,
\end{equation*}

\noindent by introducing two slack variables $\delta_{c,1}$ and $\delta_{c,2}$ that guarantee equality holds for any choice of $x_i$ satisfying the constraint.

 Let us define
\begin{equation*}
    f_c(\{x_i,\delta_{c,j}\},z_c) = \sum_{x_{i} = TRUE} (1- x_{i}) + \sum_{x_{i}= FALSE} x_{i} + \delta_{c,1} + \delta_{c,2} - (2 +z_c).
\end{equation*}

\noindent The quantity $ f_c(\{x_i,\delta_{c,j}\},z_c)$ is equal to zero when the constraints are satisfied. $3$-SAT can now be expressed as an unconstrained optimization problem by introducing $-\lambda f_c^2$ as a penalty term

    \begin{align*}\label{satform unconstrained}
    & \max\  \sum_{c \in C} (z_c - \lambda f_c(\{x_i,\delta_{c,i}\},z_c)^2)& \\
    &x_i,\ \delta_{c,i},\ z_c \in \{0,1\} .\\
\end{align*}

\noindent where $\lambda$ is a large number \cite{rockafellar2015convex}.  The variables $x_i,\delta_{c,j},$ and $z_c$ can take values of either $0$ or $1$, however, choices that violate $f_c = 0$ give penalties $ \lambda f_c(\{x_i,\delta_{c,i}\},z_c)^2 < 0$. This ensures that the optimal solution is identical to the original constrained problem. Since we introduce two $\delta$ variables and one $z$ per clause in the dualization, this formulation requires $3|C|$ ancillary qubits.  The term $-\lambda \sum_{c\in C}f_c(\{x_i,\delta_{c,j}\},z_c)^2$ must be expanded in order to determine all of the edges of the derived graph. 

The circuit depth for one layer of QAOA is the chromatic index plus one, and the chromatic index is either the maximum degree of the derived graph, or the maximum degree of the derived graph plus one. Thus, the circuit depth is either the maximum degree of the derived graph plus one or the maximum degree of the derived graph plus two.

To compute the maximum degree, let $C_{x_i} \subset C$ be the set of clauses containing $x_i$. Then the degree of the $x_i$ in the derived graph is 
\begin{equation*}
\deg(x_i) =  5|C_{x_i}| - \sum_{j, j\neq i}(|C_{x_i,x_j}|-1)
\end{equation*}
\noindent where $C_{x_i,x_j} = C_{x_i} \cap C_{x_j}$ is the set of clauses containing both $x_i$ and $x_j$. This value lies between $3|C_{x_i}|$ and $5|C_{x_i}|$. The lower bound is because $x_i$ is adjacent to both $\delta$ variables and one $z$ per clause and the upper bound is because $x_i$ can also be adjacent to distinct $x_j$ and $x_k$ in each clause.  Notice that $\deg(\delta_{i,j}) = 5$ and $\deg(z_c) = 5$. Since $\deg(x_i) \geq 5$ for some $i$, the maximum degree of the graph is $\max_i \{\deg(x_i)\}$, so the circuit depth for one layer of QAOA is either $\max_i \{\deg(x_i)\} + 1$ or $\max_i \{\deg(x_i)\} + 2$.

\subsubsection{SAT product formulation}\label{nonlinear}

Alternatively, SAT can be written as a product of monomials. To see this, note that $x_i \vee x_j \vee x_k$ is satisfied if $(x_i -1)(x_j-1)(x_k-1) = 0$. Thus, we can write SAT as the polynomial unconstrained binary optimization (PUBO) problem

\begin{equation*}
    \sum_C \prod_{x_i = TRUE_c}(1-x_i)\prod_{x_i = FALSE_c}x_i = 0,
\end{equation*}

\noindent where $x_i$ indicates if literal $x_i$ is satisfied. There are no ancillary qubits needed to write this PUBO, however expanding it does give monomials in three variables, as seen in Tab.~\ref{tab:3SATdegrees}. A straightforward implementation of an $n$ variable monomial requires an $n$-qubit gate, however, current hardware is often limited to two-qubit gates. We therefore need a method to decompose these large monomials into products of at most two variables. 

\begin{table}

    \centering
    \begin{tabular}{|c|c|c|c|c|c|}
        \hline
        Clause & Reformulation & Expansion & $\deg(x_a)$ & $\deg(x_b)$ & $\deg(x_c)$\\
        \hline
        $x_a \vee x_b \vee x_c$ & $(1-x_a)(1-x_b)(1-x_c)$ & $1-x_a-x_b-x_c+x_ax_c+x_bx_c+x_ax_b-x_ax_bx_c$ & 2 & 1 & 1 \\
        \hline
       $\backsim x_a \vee x_b \vee x_c$ & $x_a(1-x_b)(1-x_c)$ & $x_a - x_ax_b - x_ax_c+x_ax_bx_c$ & 2 & 1 & 1 \\
        \hline
        $ x_a \vee \backsim x_b \vee x_c$ & $(1-x_a)x_b(1-x_c)$ & $x_b-x_ax_b-x_bx_c+x_ax_bx_c$ & 1 & 1 & 0 \\
        \hline
       $\backsim x_a \vee \backsim x_b \vee x_c$ & $x_ax_b(1-x_c)$ & $x_ax_b - x_ax_bx_c$ & 1 & 1 & 0 \\
        \hline
        $x_a \vee \backsim x_b \vee \backsim  x_c$ & $(1-x_a)x_bx_c$ & $ x_bx_c - x_ax_bx_c$ & 0 & 0 & 0 \\
        \hline
      $\backsim x_a \vee \backsim x_b \vee \backsim x_c$ & $x_ax_bx_c$ & $x_ax_bx_c$ & 0 & 0 & 0\\
        \hline
    \end{tabular}
    \caption{
    The product reformulation and expansion of 3-SAT clauses. In this chart, we let $x_bx_c = u_{b,c}$. The last three columns count the contribution to the degree of each variable from the monomials in two vertices that are not substituted. This contribution is only added the first time the two variable monomial appears. For example, if $x_ax_b$ appears in more than one clause, we only add one to the degrees of $x_a$ and $x_b$. Since $x_bx_c = u_{b,c}$, we do not add one to the degrees of $x_b$ and $x_c$ whenever $x_bx_c$ appears in a clause. The number in each degree column is added to the degree from Thm.~\ref{thm:simplesplitdegree} to calculate the degree of vertex $x_i$ in the derived graph. }
    \label{tab:3SATdegrees}
    
\end{table}

\subsection{Example: linear and product formulations}\label{examples}

Consider the $3$-SAT problem 

\begin{example}\label{3satexample}

\begin{align*}
  & x_1 \vee x_2 \vee \backsim x_3 \\
  & x_1 \vee x_3 \vee x_4  \\
  & \backsim x_2 \vee x_4 \vee x_5  \\
  & x_1 \vee \backsim x_2 \vee x_5 .
\end{align*}

\end{example}

Labeling the indicator variable for the first clause $z_1$, the second $z_2$, and so on, we can write the objective function for this problem as 
\begin{equation*}
    \max\  \sum_{c = 1}^4 z_c \\
\end{equation*}
\noindent subject to the constraints

\begin{align*}
    & x_1+x_2+(1-x_3) \geq z_1, \\
    & x_1+x_3+x_4 \geq z_2, \\
    & (1-x_2)+x_4+x_5 \geq z_3, \\
    & x_1+(1-x_2)+x_5 \geq z_4, \\
    &x_i,\ z_c \in \{0,1\} .
\end{align*}

We consider the circuit depth for one layer of QAOA to solve this problem. First, we look at the linear characterization found in Sec.~\ref{linear}. Then, we look at the product formulation and compare the degrees of the resulting derived graphs to determine the circuit depth for one layer of QAOA from each method, assuming three qubit gates are possible. We next devise a variable-substitution method to use the product formulation in Section \ref{variablesub}, and analyze the circuit depth for a two-qubit gate implementation of the product formulation in subsection \ref{gs example subsection}.

\subsubsection{Linear Formulation}
When we use the linear formulation of the above problem, the contrapositive of the simplified dualized constraint for the first clause is

\begin{equation*}
    -x_1-x_2+x_3-z_1 \leq 0.
\end{equation*}

In order to dualize the constraint, we add in two slack variables, $\delta_{1,1}$ and $\delta_{1,2}$ and change the inequality to equality to obtain

\begin{equation*}
    -x_1-x_2+x_3-z_1+ \delta_{1,1} + \delta_{1,2} = 0.
\end{equation*}

\begin{figure}
\centering
\includegraphics{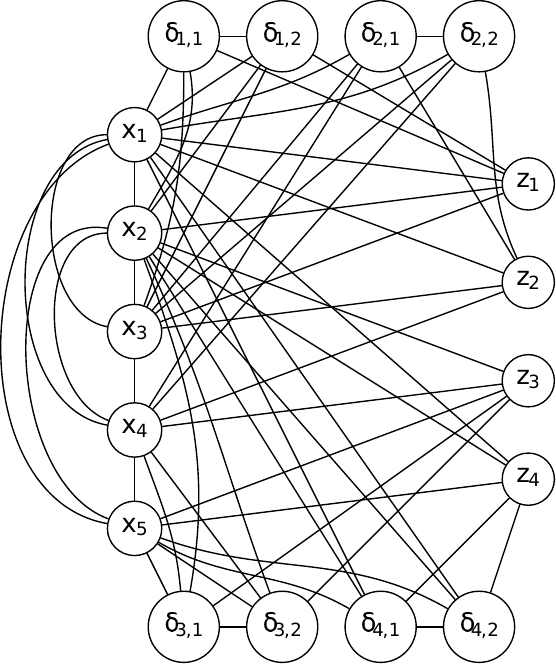}
\caption{The derived graph for the dualization of the linear formulation of Ex.~\ref{3satexample}. The maximum degree of the graph is $13$, the chromatic index of this graph is either $13$ or $14$, and the circuit depth for one layer of QAOA is either $14$ or $15$.}
\label{linearderivedgraph}
\end{figure}

Upon squaring both sides, we get all terms of the form $x_ix_j$, $x_iz_1$, $x_i\delta_{1,k}$, $z_1\delta_{1,k}$, and $\delta_{1,1}\delta_{1,2}$ where $i \neq j$, $i,j \in \{1,2,3\}$ and $k \in \{1,2\}$. The other clauses are handled similarly and the derived graph is found in Fig.~\ref{linearderivedgraph}. The  degree of the vertex that is maximal in the derived graph is thirteen, so the circuit depth is either fourteen or fifteen for one layer of QAOA.

\subsubsection{Product Formulation}
Instead of solving the linear formulation, each clause can be rewritten as a product of combinations of $x_i$ and $(1-x_j)$. For Ex.~\ref{3satexample}, the first statement holds if and only if $(1-x_1)(1-x_2)x_3$ holds. Upon expanding, we get the expression $x_3 - x_1x_3 - x_2x_3 + x_1x_2x_3$. The other clauses are handled similarly. If there are three-qubit gates, this method requires no ancillary qubits. 

Expanding the product formulation of a single $3$-SAT clause results in a monomial that contains three variables and possibly one or more monomials in two variables. The variables that occur in degree two monomials are all contained in the three variable monomial. Since each monomial represents a gate that acts on the qubits contained in the monomial, we need only implement the three-qubit gates from each clause. Thus, we can eliminate all edges from the derived graph and are left with a hypergraph, Fig.~\ref{3SATProductDerivedGraph}. The circuit depth is equal to the chromatic index of the hypergraph plus one.


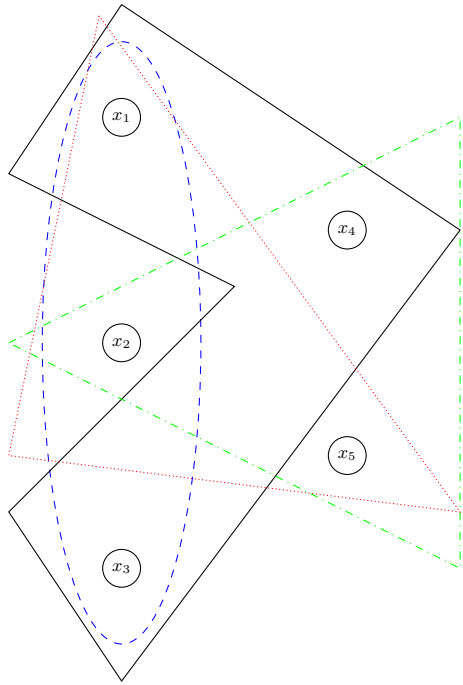
\begin{figure}
 \centering
 \begin{tikzpicture}[scale=3]
\begin{scope}[every node/.style={scale=.75,circle,draw}]
    \node (A) at (0,4) {$x_1$};
    \node (B) at (0,3) {$x_2$};
	\node (C) at (0,2) {$x_3$}; 
	\node (D) at (1,3.5) {$x_4$};
	\node (E) at (1,2.5) {$x_5$}; 

\end{scope}




\draw [blue, dashed] (0,3) ellipse (10pt and 38pt);
\draw (0,4.5) -- (1.5,3.5) -- (0,1.5)--(-.5,2.25)--(0.5,3.25)--(-.5,3.75)--(0,4.5);
\draw [green, dashdotted] (-.5,3) -- (1.5,4) -- (1.5,2)--(-.5,3);
\draw [red, densely dotted] (-.1,4.45) -- (-.5,2.5) -- (1.5,2.25)--(-.1,4.45);

\end{tikzpicture}
\caption{The derived graph for the product formulation of Ex.~\ref{3satexample}. There are four hyperedges: $x_1x_2x_3$ in dashed blue, $x_1x_3x_4$ in black, $x_2x_4x_5$ in dashdotted green, and $x_1x_2x_5$ in densely dotted red. The chromatic index of the graph is four and the circuit depth is five.}
\label{3SATProductDerivedGraph}
\end{figure}

\section{Global Variable Substitution}\label{variablesub}

The previous section assumes that three-qubit gates are possible on the hardware, however, that may not always be the case. Thus, we examine how to decompose monomials in three or more variables via a process called global variable substitution (GVS). In order to substitute variables, we require constraints on the problem to ensure the substitution is valid. We then eliminate the constraints via dualization and can determine the circuit depth of one layer of QAOA. 

Let us define a substitution of size $s$ as one in which $s$ variables are combined into one e.g., $x_1x_2...x_s = u_{1,2,...,s}$. Throughout the remainder of this paper we use the boldface variable $\bf j$ to refer to lists of indices, ${\bf{j}} =(j_1,j_2, ...j_m)$, where $j_i \in \mathbb{N} \ \forall i \leq m$.  A substitution is denoted $u_{\bf j} = x_{j_1}x_{j_2}...x_{j_m}$.  When we wish to refer to specific substitutions, we will specify indices of the variables, e.g. $u_{i,j}=x_ix_j$.

The goal of this work is to write an $n$ variable monomial as a product of two variables. One variable replaces $s$ of the variables of the original monomial, and another variable is substituted for the remaining $n-s$. Since the order of the variables in each monomial is not important, we assume without loss of generality that the first $s$ variables are substituted for one variable and the last $n-s$ for the other. In this formulation, we allow no overlap in the substitutions, i.e. if $x_i$ is a variable in the substitution of size $s$, then it is not a variable in the substitution of size $n-s$.

In order to substitute a new variable $u_{1,...,s}$ for $x_1...x_s$, we add in the constraints 

\begin{align}\label{theglobalconstraints}
   u_{\bf j} & =  \sum_{ i \in[|u_{\bf j}|]} x_i - \sum_{i =1}^{m_{u_{\bf j}}} \delta_{c,k} - (|u_{\bf j}|-1)\\
    u_{\bf j} & = x_i - \delta_{c,m_{u_{\bf j}}+k}  & \forall i \in [|u_{\bf j}|] \label{eq:globalsubconstraint}
\end{align}
\noindent where $|u_{\bf j}|$ denotes the number of variables $u_{ \bf j}$ replaces and $m_{u_{\bf j}} \in \{1, ... , s\}$, depending on if $x_i$ takes the value of $0$ or $1$ for each $x_i$. The first constraint ensures that if $x_i = 1 \ \forall \  i$, then $u_{1,2,...,s} =1$, while the second guarantees that if $x_i = 0$ for some $i$, then $u_{1,2,...s} = 0$. Note that $|u_{\bf j}| = s_j$, so the two can be used interchangeably. After dualizing these constraints, we are able to determine the degree of each vertex in the derived graph. To facilitate the discussion, in Theorem 1 we focus only on the vertices and edges associated with the dualized substitution terms, which form a sub-graph of the total derived graph.  We then show how these are related to the total derived graph and describe how to compute the circuit depth for one layer of QAOA in Section \ref{edgecounting}.

\subsection{Degree of substituted variables}\label{subdegree}
We now give a formula for the degrees of vertices $\delta_{c,k}$ and $u_{\bf j}$, as well as determine the contribution of each substitution to the degree of $x_i$. These degrees, in conjunction with the $x_ix_j$ terms, can be used to compute the circuit depth, as shown explicitly in the example of subsection \ref{gs example subsection}.    
\begin{theorem}\label{thm:simplesplitdegree}
Let $U$ denote the set of all substitutions made. Let us denote the set of substitutions containing $x_i$ as $U_i$. For $u_{ \bf j} \in U$ we let $|u_{\bf j}|$ denote the number of variables substituted in substitution $u_{\bf j}$. The number of edges incident to each vertex in the derived graph due to the substitution is denoted $\deg_s(v)$ and is 

\begin{align*}
    & \deg_s(\delta_{c,k})  \in \{2, |u_{\bf j}|+m_{u_{\bf j}}\} \\
    & \deg_s(x_{i}) = \sum_{u_{\bf j} \in U_i} \big[|u_{\bf j}| + m_{u_{\bf j}} + 1\big] - \sum_{p \in [n]\setminus i}z_{i,p}\big[|U_i \cap U_p|-1 \big] + |C_{x_i}'|\\
    & \deg_s(u_{\bf j})  = 2|u_{\bf j}| + m_{u_{\bf j}} + |C_{u_{\bf j}}| \\
\end{align*}

\noindent where $m_{u_{\bf j}}$ denotes the number of $\delta_{c,k}$ variables in Eq.~\eqref{theglobalconstraints} for the substitution $\bf j$, $z_{i,p}$ is an indicator variable denoting whether or not $U_i \cap U_p = \{\emptyset\}$, $| C_{x_i}'|$ refers to the number of monomials containing $x_i$ which do not have any $u_s$ containing $x_i$ substituted, and $|C_{u_{\bf j}}|$ refers to the number of monomials with more than two variables using the substitution $u_{\bf j}$. 
\end{theorem}
\begin{proof}
First, we will consider the degree of each $\delta_{c,k}$ by counting the number of terms that contain $\delta_{c,k}$ in each dualized constraint. Note that in the first constraint, there are $1 + |u_{\bf j}| + m$ variables total, so $|u_{ \bf j}|+m$ of these are multiplied times a single $\delta_{c,k}$ term upon dualization. Thus, the degree of $\delta_{c,k}$ in this constraint is $|u_{\bf j}|+m$. In each of the following constraints, note that there are three variables: $u_{\bf j}$, $x_i$ and $\delta_{c,k}$. Thus, the degree of $\delta_{c,k}$ from these constraints is $2$. Since the $\delta_{c,k}$ terms in each constraint are different, they do not add, so $\deg_s{\delta_{c,k}} \in \{2, |u_{\bf j}|+m_{u_{\bf j}}\}$.

Next, we consider the degree of each $x_i$ in the graph derived from the dualization. In the first constraint, there are $|u_{\bf j}|+m_{u_{\bf j}}$ terms containing each $x_i$ and there is exactly one constraint aside from the first that contains $x_i$. This constraint only has two terms containing $x_i$, but the contribution to the degree from this is one since the edge $u_{\bf j}x_i$ already exists in the graph from the first constraint. Now, we must consider how many $u_{\bf j}$ contain $x_i$ as a variable. If there is more than one substitution containing $x_i$, there are multiple constraints that have the form of the first one. Each of those constraints contain new $u_{\bf j}$ variables and new $\delta_{c,k}$ variables since the substitutions are different. Thus, the only double counting that can happen is in products of $x_ix_j$, since these are the only other edges possible in the graph derived from the dualization. We must subtract the number of times each of these terms occurs in each constraint except for one. Additionally, we must count the clauses that contain $x_i$ but in which $x_i$ is not contained in a substitution, since it will be incident to the substitution in the derived graph. We denote the number of these clauses as $|C_{x_i}'|$. Then, $\deg_s(x_{i}) = \big[\sum_{u_{\bf j} \in U_i}|u_{\bf j}| + m_{u_{\bf j}} + 1\big] - \sum_{p \in [n]\setminus i}z_{i,p}\big[|U_i \cap U_p|-1 \big] + |C_{x_i}'|$. 

Finally, we need to consider the degree of each substituted variable, $u_{\bf j}$. The first constraint contributes $|u_{\bf j}| + m$ to the degree and the other $|u_{\bf j}|$ constraints contribute $1$ to the degree since the first constraint accounts for the $u_{\bf j}x_i$ terms. Finally, each $u_{\bf j}$ is substituted into a clause, so is multiplied by the variable not contained in the substitution. We denote the number of clauses into which $u_{\bf j}$ is substituted as $|C_{u_{\bf j}}|$. Thus,  $\deg_s{u_{\bf j}} = 2|u_{\bf j}| + m_{u_{\bf j}} + |C_{u_{\bf j}}|$.
\end{proof}

The derived graph for the problem consists of variables and edges induced by the substitution, as well as variables and edges $x_ix_j$ that exist in the original problem. The variables specific to the substitutions, $\delta_{c,k}$ and $u_{\bf j}$, have degrees determined solely by Thm.~\ref{thm:simplesplitdegree}. The chromatic index of the derived graph is the circuit depth for one layer of QAOA, and the maximum of the degrees from Thm.~\ref{thm:simplesplitdegree} plus one provides a lower bound for this quantity. We show how to compute the exact circuit depth in Section \ref{edgecounting}, including the extra edge terms $x_ix_j$.


\subsection{$3$-SAT}

Notice that Thm.~\ref{thm:simplesplitdegree} reduces to 

\begin{align*}
    & \deg(\delta_{c,k})  \in \{2, 3\} \\
    & \deg(x_{i}) = \big[\sum_{u_{\bf j} \in U_i} 4 \big] + |C_{x_i}'|\\
    & \deg(u_{\bf j})  = 5 + |C_{s_j}| \\
\end{align*}

\noindent for $3$-SAT since $|u_{\bf j}|=2$, $z_{i,p}|U_i \cap U_p| \in \{0,1\}$, and $m=1$. This is the degree for each variable due to the substitution. In order to determine the total degree, we need to look at two variable terms in the expansions and add one to the degree for each unique term.

\subsubsection{Counting edges in $3$-SAT that are not due to substitutions}\label{edgecounting}

To compute the circuit depth for one layer of QAOA, we must compute the maximum vertex degree in the derived graph. We focus here on the example of 3-SAT. A similar approach can be used to decompose other classes of combinatorial optimization problems.
In order to calculate the total degree of a vertex in the derived graph from $3$-SAT, we need to add the degrees of a vertex due to a substitution $u_{\bf j}$, from Thm.~\ref{thm:simplesplitdegree}, to the degree from two vertex monomials that were not substituted.  Since $\delta_{c,k}$ and $u_{i,j}$ are introduced in order to make the substitutions, their degrees are exactly the quantity in Thm.~\ref{thm:simplesplitdegree}. 

In order to count the number of edges induced by pairs that are not substituted, we define $P$ as the set of all two variable monomials that result from the expansion of the product formulation of each clause in a $3$-SAT instance. These two variable monomials will be denoted $s_{i,j}$. For example, if we have the clauses $(1-x_1)x_2x_3 = x_2x_3-x_1x_2x_3$ and $(1-x_2)x_4(1-x_5) = x_4 - x_2x_4 - x_4x-5 + x_2x_4x_5$ and the substitutions $u_{1,3}$ and $u_{4,5}$ are made, in this case, $P = \{x_2x_3, x_2x_4\}$. The subset of $P$ containing a vertex $a$ is denoted $P_a$. Here, $P_2 = P$, since $x_2$ is in each monomial, $P_3 = \{x_2x_3\}$, and  $P_4 = \{x_2x_4\}$. Let $||P|| = (|P_1|, |P_2|, ... , |P_n|)$. 

These sets can be used to determine the number of edges incident to a vertex $x_i$ that are from the two variable terms in the expansions of each clause, which we denote $\deg_e(x_i)$. Let us fix variable $x_i$. Since $|P_i|$ counts the number of two variable terms containing $x_i$, it is added to Thm.~\ref{thm:simplesplitdegree}. If one of the two variable terms from the expansion of the clauses is substituted, i.e. if $x_ix_j \in P_i$ and $u_{i,j}$ is a substitution, then adding $|P_i|$ to the degree double counts the edge $x_ix_j$. Thus, the number of substitutions that appear in $P_i$ need to be subtracted. If $S_i$ is the set of substitutions containing variable $x_i$ and $y_s$ is an indicator variable that determines if the substitution $s_{i,j}$ was used, the total degree of $x_i$ is 
\begin{equation*}
    \deg(x_{i}) = 4 \: \sum_{s_{i,j} \in S_{i}} y_{s_{i,j}} - \sum_{s_{i,j} \in P_{i}} y_{s_{i,j}} + |P_{i}| +   |C_{x_i}'|= \deg_s(x_i) + \deg_e(x_i),
\end{equation*}
\noindent where $\deg_e(x_i) = - \sum_{s_{i,j} \in P_{i}} y_{s_{i,j}} + |P_{i}|$. 

The QAOA circuit depth is then given as $\max \{\deg(\delta_{c,k}), \deg(x_i), \deg(u_{i,j})\}$, where

 \begin{align}
    & \deg(\delta_{c,k})  \in \{2, 3\}\label{eq:3SATDegsdelta} \\
    & \deg(x_{i}) = 4 \: \sum_{s_{i,j} \in S_{i}} y_{s_{i,j}} - \sum_{s_{i,j} \in P_{i}} y_{s_{i,j}} + |P_{i}| +   |C_{x_i}'|\label{eq:3SATDegs_v_a}\\
    & \deg(u_{{i,j}})  = 5 + |C_{u_{i,j}}|, \label{eq:3SATDegs_uij}
\end{align} 
for $3$-SAT. Note that $\delta_{c,k} \in \{2,3\}$ since the size of each substitution $|u_{i,j}| = 2$ and $m=n-|u_{i,j}|=3-2 = 1$, which also simplifies $\deg(u_{i,j})$. Also note that $U_i \cap U_p \in\{0,1\} \ \forall \ i \neq p, \ i,p\in [n]$, which simplifies $\deg(x_i)$.

\subsubsection{Example: Global variable substitution} \label{gs example subsection}

To give an example of the GVS method, we will apply it to the product formulation of Ex.~\ref{3satexample}. We want to decompose the three-variable terms into two-variable terms by substituting a new variable to represent the product of the variables. In order to do so, we first choose substitutions and then create the derived graph. 
For Ex.~\ref{3satexample}, we substitute $x_1x_3 = u_{1,3}$ and $x_2x_5 = u_{2,5}$. Thus, the first clause, which can be written as $x_3 - x_2x_3 - x_1x_3 + x_1x_2x_3$ is equal to $x_3 - x_2x_3 - u_{1,3} + u_{1,3}x_2$. With the substitution $x_1x_3 = u_{1,3}$, the edge $x_1x_3$ is already accounted for in the first constraint. We then need only add the edge between $x_2$ and $x_3$ to the derived graph since the monomial $x_2x_3$ exists in the expansion of the first clause. Note that each clause can be decomposed similarly and will end up with a sum of terms, at least one of which contains three variables, with the others containing one or two depending on the number of $(1-x_i)$ expressions. The other edges that need to be added due to the two variable terms in the expansions are $x_1x_4$, $x_3x_4$, $x_1x_2$ and $x_2x_4$. See Tab.~\ref{tab:3SATdegrees} for the contribution of unsubstituted monomials from a particular clause to $\deg_e(x_i)$ for a generic substitution $x_bx_c = u_{b,c}$. The increase in degree ranges from $0$ to $2$, with the maximum addition being $k-1$ for general $k$-SAT.

Next, we list each constraint of the form Eqs.~\eqref{theglobalconstraints}-\eqref{eq:globalsubconstraint} and dualize them. For the first clause $c=1$ in the example,

\begin{align*}
    u_{1,3} & =  x_1 + x_3 + \delta_{1,1} \\
    u_{1,3} & = x_1 + \delta_{1,2}   \\
    u_{1,3} & = x_3  + \delta_{1,3} 
\end{align*}

\noindent When dualizing these constraints, the terms

\begin{align}
    (u_{1,3} - x_1 - x_3 - \delta_{1,1})^2 & = u_{1,3}+x_1+x_3+\delta_{1,1} -u_{1,3}x_1- u_{1,3}x_3 - u_{1,3} \delta_{1,1} +x_1 \delta_{1,1} + x_3  \delta_{1,1} + x_1x_3 \label{u13first}\\
    (u_{1,3} - x_1 - \delta_{1,2})^2  & = u_{1,3}+x_1+\delta_{1,2} -u_{1,3}x_1 - u_{1,3} \delta_{1,2}+  x_1 \delta_{1,2} \label{u13second}\\
    (u_{1,3} - x_3  - \delta_{1,3})^2 & = u_{1,3}+x_3+\delta_{1,3} -u_{1,3}x_3 - u_{1,3} \delta_{1,3}+  x_3 \delta_{1,3} \label{u13third}
\end{align}
\noindent are added to the objective function, up to constant $\lambda$, and similar terms are added for the  $u_{2,5}$ substitution. Since all variables, $v$, in the equations above have the value zero or one, $v^2 = v$. The entire derived graph can be seen in Fig. \ref{3SATGlobalderivedgraph}. Note that the degrees match the theorem plus the number of edges induced by pairs that are not substituted. The largest degree vertex is $x_2$, which has degree eight, so the circuit depth for one layer of QAOA is either nine or ten. This is less than the circuit depth for the linear $3$-SAT formulation and requires fewer ancillary qubits.

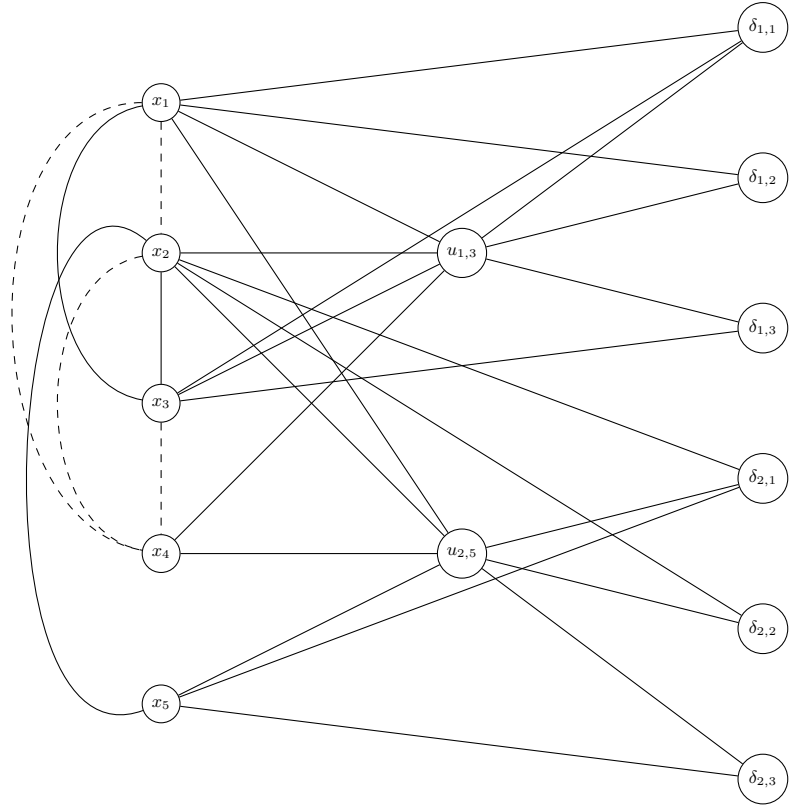
\begin{figure}
 \centering
 \begin{tikzpicture}[scale=2]
\begin{scope}[every node/.style={scale=.75,circle,draw}]
    \node (A) at (0,4) {$x_1$};
    \node (B) at (0,3) {$x_2$};
	\node (C) at (0,2) {$x_3$}; 
	\node (D) at (0,1) {$x_4$};
	\node (E) at (0,0) {$x_5$}; 
	\node (F) at (2,3) {$u_{1,3}$};
	\node (G) at (2,1) {$u_{2,5}$}; 
	\node (H) at (4,4.5) {$\delta_{1,1}$};
	\node (I) at (4,3.5) {$\delta_{1,2}$};
	\node (J) at (4,2.5) {$\delta_{1,3}$};
	\node (K) at (4,1.5) {$\delta_{2,1}$};
	\node (L) at (4,0.5) {$\delta_{2,2}$};
	\node (M) at (4,-.5) {$\delta_{2,3}$};
\end{scope}

\draw  [dashed](A) -- (B);
\draw  (A) to[out=190,in=170] (C);
\draw  [dashed] (A) to[out=180,in=170] (D);
\draw  (A) -- (F);
\draw  (A) -- (G);
\draw  (A) -- (H);
\draw  (A) -- (I);

\draw  (B) -- (C);
\draw  [dashed](B) to[out=190,in=170] (D);
\draw  (B) to[out=140,in=200] (E);
\draw  (B) -- (F);
\draw  (B) -- (G);
\draw  (B) -- (K);
\draw  (B) -- (L);

\draw  [dashed](C) -- (D);
\draw  (C) -- (F);
\draw  (C) -- (H);
\draw  (C) -- (J);

\draw  (D) -- (F);
\draw  (D) -- (G);

\draw  (E) -- (G);
\draw  (E) -- (K);
\draw  (E) -- (M);

\draw  (F) -- (H);
\draw  (F) -- (I);
\draw  (F) -- (J);

\draw  (G) -- (K);
\draw  (G) -- (L);
\draw  (G) -- (M);

\end{tikzpicture}
\caption{The derived graph for solving Ex.~\ref{3satexample} using GVS with substitutions $x_1x_3 = u_{1,3}$ and $x_2x_5 = u_{2,5}$. The vertices of this graph are the variables from the problem, $x_i$ for $i \in [5]$, along with the substitution variables, $u_{1,3}$ and $u_{2,5}$, and the three $\delta_{c,k}$ variables needed per substitution. The solid edges correspond to the edges induced by the substitution, which correspond to products of two monomials in the Eq.~\eqref{u13first}-\eqref{u13third}. The dashed edges represent the two variable terms from the expansion of the $3$-SAT clauses that are not involved in substitutions. The maximum degree of this graph is $8$, thus the circuit depth for one layer of QAOA is $9$ or $10$.}
\label{3SATGlobalderivedgraph}
\end{figure}

\section{Optimizing Global Variable Substitutions for $3$-SAT}\label{optimization}

As seen above, the number of substitutions impacts the degrees of the vertices. We want to minimize both the number of ancillary qubits and the maximum degree of the derived graph. This problem is difficult for large gates, but in this section, we explore how to solve this problem with gates with three or fewer variables such as in $3$-SAT.

\subsection{Global Variable Substitution Integer Program} \label{IP_Formulation}

Each substitution introduces new variables that correspond to ancillary qubits and impacts the degree of each vertex in the derived graph. In order to determine which of the possible combinations of substitutions are feasible, we create a graph $G = (V(G), E(G))$, which we call the \textit{covering graph}. There are two sets of vertices in this graph. One set is the set of all $2$-sets representing all possible two-variable substitutions $u_{i,j}=x_ix_j$. We call this set of vertices $S$ and denote the pairs $x_ix_j$ in $S$ as $s_{i,j}$. Each $s_{i,j}$ is a possible substitution, but not each $s_{i,j}$ will result in a substitution $u_{i,j}$. The other set is the one containing all three variable terms from the expansion, which we will call the \textit{expansion 3-set}, $ES_{3}$. Edges are placed between vertices $g$ and $h$ if the variables in the label of $g$ are a proper subset of the variables in the label of $h$. There are no edges between any sets of the same size, making the graph bipartite. See Fig.~\ref{3SATcoveringgraph} for an example, where $S$ is the left set of vertices and $ES_3$ is the right set for Example 1.

 The nature of this problem naturally lends itself to a set covering or bipartite matching style formulation as each substitution $u_{i,j}$ can cover, or be matched with, any clause containing the literals in that substitution. We have chosen to use a set covering problem formulation, i.e. we want to find a subset of $S$ such that each vertex in $ES_3$ is incident to a vertex in the subset. Since covering a vertex in $ES_{3}$ more than once leads to unnecessary ancillary qubits and increases the degree, we need to add constraints to ensure that each clause $c \in C$ is covered by exactly one substitution $u_{i,j}$. Notice that more than one covering may be possible as seen in Fig.~\ref{3SATcoveringgraph}. 
 
The objective of $3$-SAT is to minimize the maximum degrees found in Eq.~\eqref{eq:3SATDegsdelta}-~\eqref{eq:3SATDegs_uij}, 
\begin{align*}
\min \:  \max\{\deg(\delta_{c,k}),\deg(x_{i}),\deg(u_{i,j})\}
\end{align*}

\noindent which are derived from the GVS method. The covering graph developed in the previous paragraphs is used to aid in the minimization process. The integer program formulation of this problem is subject to the set covering constraints. Each clause $x_{i}x_{j}x_{k} = c$ is covered by a pair $s_{i,j} = x_ix_j$ and has a variable $x_{k}$ that is not in a substitution. The indicator variable $z(c,s_{i,j},x_k)$ is defined as

\begin{align*}
  z(c,s_{i,j},x_k) = & \begin{cases}
  \; 1 \text{ if } \textup{covering} \: (s_{i,j},x_k) \: \text{is selected to cover clause} \: c \\ 
  \;0 \text{ else } 
\end{cases}\\
\end{align*}

The full integer program formulation to minimize the maximum degree vertex in the derived graph, $obj$, is as follows:

\begin{align}
 \min \ &  obj + \frac{1}{10|C|}*\sum_{s \in S} y_{s} \label{obj}\\
s.t.\ & 4 \: \sum_{s_{i,j} \in S_a} y_{s} - \sum_{s_{i,j} \in P_a} y_{s_{i,j}} + |P_{a}| + \sum_{c \in C} \sum_{\{s_{i,j} \in S\ |\ a \notin s_{i,j}\}}  z(c,s_{i,j},x_a) \leq obj & \forall a \in V \label{maxDegv}\\
& 5 + \sum_{c \in C} z(c,s_{i,j},x_k) \leq obj &  \forall s \in S \label{maxDegS} \\
& \sum_{s \in S} z(c,s_{i,j},x_k) = 1 & \forall c \in C \label{setCover}\\
& \sum_{c \in c} z(c,s_{i,j},x_k) - y_{s_{i,j}} \leq 0 &  \forall s  \in S \label{subIndicator}\\
& obj \in \mathbb{Z}^+ \label{int_obj}\\
&y_{s_{i,j}}, z(c,s_{i,j},x_k) \in \{0,1\} \label{binary}
\end{align}

\noindent 
Eq.~\eqref{obj} is our objective function. We have added the penalty $\frac{1}{10|C|}*\sum_{s_{i,j} \in S} y_{s_{i,j}}$ in Eq.~\eqref{obj} which minimizes the total number of unique substitutions to help minimize the number of ancillary qubits.  
Eq.~\eqref{maxDegv} indicates that the maximum degree of each vertex must be less than or equal to the objective value. The last term of this constraint counts the number of clauses that contain variable $x_a$ but in which $x_a$ is not substituted. This is equivalent to $|C_{x_a}'|$ in Thm.~\ref{thm:simplesplitdegree} and the example. Similarly, Eq.~\eqref{maxDegS} indicates that the maximum degree of any substituted variable must be less than or equal to the objective value. These two constraints effectively allow us to minimize the maximum degree of the graph. Note, the degree of the slack variables are not included in this formulation as they are either $2$ or $3$ and will never yield the maximum degree in a problem of sufficient size. Eq.~\eqref{setCover} is our set covering constraint. This asserts that each $c \in C$ is required to be covered by exactly one covering $s_{i,j},x_k$. Eq.~\eqref{subIndicator} asserts that if at least one $c \in C$ substitutes the pair $s_{i,j}$, $y_{s_{i,j}}$ is set to 1. 

\subsubsection{IP Formulation and Solution for Example 1}
Let us examine Ex.~\ref{3satexample}, 
\begin{align*}
  & x_1 \vee x_2 \vee \backsim x_3 \\
  & x_1 \vee x_3 \vee x_4  \\
  & \backsim x_2 \vee x_4 \vee x_5  \\
  & x_1 \vee \backsim x_2 \vee x_5 .
\end{align*}

\noindent The reformulation and expansion of each of these four 3-SAT clauses yields a three variable monomial $c$, two variable monomials $s_{i,j} \in P$, and potential coverings $(s_{i,j}, x_{k})$.The first statement holds if and only if $(1-x_1)(1-x_2)x_3 = x_3 - x_1x_3 - x_2x_3 + x_1x_2x_3$ holds. The other clauses, when expanded, give expressions $1+x_1x_4+x_3x_4-x_1x_3x_4$, $x_2 - x_2x_4 - x_2x_5 + x_2x_4x_5$, and $x_2 - x_1x_2 - x_2x_5 +x_1x_2x_5$. The two variable terms are $s_{i,j} \in P$. Each $c \in C$ has $\binom{3}{2}$ pairs that can cover it. In this example, we have five variables and four clauses. A summary of the clauses, their corresponding $P$ and potential covering are found in Tab.~\ref{tab:IP_Example}.

 \begin{table}

    \centering
    \begin{tabular}{|c|c|c|c|c|c|}
    \hline
     & c ($x_{i}x_{j}x_{k}$) & $ s_{i,j}\in P$   & Covering 1  & Covering 2 & Covering 3 \\
    \hline
    c0 & $x_{1}x_{2}x_{3}$           & $s_{1,3}, s_{2,3}$        & $(s_{1,2},x_{3})$ &  $(s_{1,3},x_{2})$   &   $(s_{2,3},x_{1})$ \\
    \hline
    c1 & $x_{1}x_{3}x_{4}$           & $s_{1,3},s_{1,4},s_{3,4}$ & $(s_{1,3},x_{4})$ &  $(s_{1,4},x_{3})$   &   $(s_{3,4},x_{1})$ \\
    \hline
    c2 & $x_{2}x_{4}x_{5}$           & $s_{2,4},s_{2,5}$         & $(s_{2,4},x_{5})$ &  $(s_{2,5},x_{4})$   &   $(s_{4,5},x_{2})$ \\
    \hline
    c3 & $x_{1}x_{2}x_{5}$           & $s_{1,2},s_{2,5}$        & $(s_{1,2},x_{5})$ &  $(s_{1,5},x_{2})$   &   $(s_{2,5},x_{1})$ \\
    \hline
    \end{tabular}
    \caption{This table shows the result of the reformulation of each $3$-SAT clause including the three variable monomials to be covered $c$, the two variable monomial pairs from the expansion of each clause $s \in P$, and each of the three potential coverings of $c$, $(s_{i,j},x_{k})$.}
    \label{tab:IP_Example}
\end{table}

\begin{align*}
    &S = \{s_{1,2},\:s_{1,3},\:s_{1,4},\:s_{1,5},\:s_{2,3},\:s_{2,4},\:s_{2,5},\:s_{3,4},\:s_{4,5}\} \\
    &P = \{s_{1,2},\:s_{1,3},\:s_{1,4},\:s_{2,3},\:s_{2,4},\:s_{2,5},\:s_{3,4}\} \\
    &||P|| = (3,4,3,3,1).
\end{align*}

\noindent After determining these sets, we are able to formulate our constraints. First  we add a constraint for each literal vertex $v_{a}\in V$ in the form of Eq.~\eqref{maxDegv}. For $a=1$: 
\begin{align*}
& \deg(x_{1}) = 3y_{s_{1,2}} + 3y_{s_{1,3}} +3y_{s_{1,4}} + 4y_{s_{1,5}} + z(c_{0},s_{2,3},x_{1}) + z(c_{1},s_{3,4},x_{1}) + z(c_{3},s_{2,5},x_{1}) + 3 \leq obj
\end{align*}

\noindent Next, we add a constraint for each unique substitution $u_{s}$ corresponding to an $s \in S$ in the form of Eq.~\eqref{maxDegS}. For  $s = (x_{1},x_{2})$: 

\begin{align*}
& \deg(u_{1,2}) = 5 + z(c_{0},s_{1,2},x_{3}) + z(c_{3},s_{1,2},x_{5})  \leq obj
\end{align*}

\noindent For each $c \in C$ we add a constraint in the form of Eq.~\eqref{setCover} to assert that a clause $c$ must be covered by exactly one of its three potential coverings found in Tab.~\ref{tab:IP_Example}. For clause $c_0$: 

\begin{align*}
& z(c_0,s_{1,2},x_{3}) + z(c_0,s_{1,3},x_{2}) + z(c_0,s_{2,3},x_{1}) = 1
\end{align*}

\noindent Finally, for each covering containing a pair $s_{i,j}$, we add a constraint in the form of Eq.~\eqref{subIndicator}. This asserts that if any covering containing $s_{i,j}$ is selected for a substitution, $x_{s_{i,j}}$ is forced to 1. For $s = (x_{1},x_{2})$ we add the constraints: 

\begin{align*}
& z(c_0,s_{1,2},x_{3}) - y_{s_{1,2}} \leq 0 \\
& z(c_3,s_{1,2},x_{5}) - y_{s_{1,2}} \leq 0 
\end{align*}

The solution to this IP is 
\begin{align*}
   &  y_{1,3}, y_{2,5} = 1 \\
   &  z(c_0,s_{1,3},x_{2}),z(c_1,s_{1,3},x_{4}),z(c_2,s_{2,5},x_{4}), z(c_3,s_{2,5},x_{1})  = 1, \\
\end{align*}
\noindent which indicates that the optimal solution is to substitute pairs $s_{1,3}$ and $s_{2,5}$. Using these substitutions, $\Delta_{G} = 8$ and the degree of each vertex in the derived graph is
\begin{align*}   
    \deg(x_{1}) = 7 \;\;
    \deg(x_{2}) = 8 \;\;
    \deg(x_{3}) = 6 \;\;
    \deg(x_{4}) = 5 \;\;
    \deg(x_{5}) = 4 \;\;
    \deg(u_{1,3}) = 7 \;\;
    \deg(u_{2,5}) = 7 \;\;
\end{align*}
\noindent so the circuit depth is at most nine.
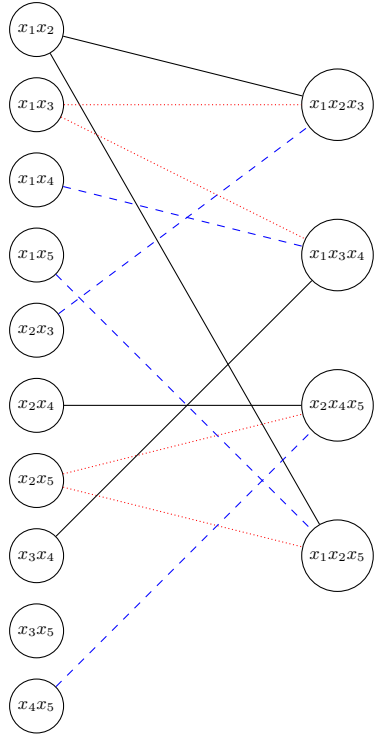
\begin{figure}
 \centering
 \begin{tikzpicture}[scale=2]
\begin{scope}[every node/.style={scale=.75,circle,draw}]
    \node (A) at (0,4.5) {$x_1x_2$};
    \node (B) at (0,4) {$x_1x_3$};
	\node (C) at (0,3.5) {$x_1x_4$}; 
	\node (D) at (0,3) {$x_1x_5$};
	\node (E) at (0,2.5) {$x_2x_3$}; 
	\node (F) at (0,2) {$x_2x_4$};
	\node (G) at (0,1.5) {$x_2x_5$}; 
	\node (H) at (0,1) {$x_3x_4$};
	\node (I) at (0,.5) {$x_3x_5$};
	\node (J) at (0,0) {$x_4x_5$};
	\node (K) at (2,4) {$x_1x_2x_3$};
	\node (L) at (2,3) {$x_1x_3x_4$};
	\node (M) at (2,2) {$x_2x_4x_5$};
	\node (N) at (2,1) {$x_1x_2x_5$};
\end{scope}

\draw  (A) -- (K);
\draw  (A) -- (N);
\draw  [red, densely dotted] (B) -- (K);
\draw  [red, densely dotted] (B) -- (L);
\draw   [blue, dashed] (C) -- (L);
\draw   [blue, dashed] (D) -- (N);
\draw   [blue, dashed] (E) -- (K);
\draw   (F) -- (M);
\draw   [red, densely dotted](G) -- (M);
\draw   [red, densely dotted](G) -- (N);
\draw   (H) -- (L);
\draw   [blue, dashed] (J) -- (M);

\end{tikzpicture}
\caption{The covering graph for Ex.~\ref{3satexample}. We have emphasized all three different coverings by coloring the edges differently for each covering. They are the red dotted, blue dashed, and black solid edges.}
\label{3SATcoveringgraph}
\end{figure}
\subsection{Heuristic Approximation}

The previous integer programming approach will provide an optimal solution. However, depending on the $3$-SAT instance, it might be a very difficult problem to solve. As an alternative, we have developed a heuristic to approximate it. Greedy algorithms are a common way to approximate large scale set covering problems \cite{grossman1997computational}. Generally, this involves a greedy selection of elements that cover the greatest amount of sets until all sets are covered. We have chosen to use a similar approach as the degree of the graph derived using GVS is directly impacted by the number of substitutions made. Therefore, we expect by selecting the coverings with the highest degree in the covering graph, we will be able to minimize the total number of substitutions made and obtain a locally minimal solution. Let $U$ be the set of uncovered clauses. Let $K$ be the set of covered clauses. At the beginning of the iteration, $K$ is empty and $U = C$, where $C$ is the set of all clauses in the $3$-SAT instance. The greedy algorithm is described in Alg.~\ref{alg:covering}. Often in step 2, multiple $s_{i,j}$ pairs may cover the same number of clauses. To break these ties, we randomly select a pair $s_{i,j}$ of the current largest degree in the covering graph.

\begin{algorithm}[H]
	\caption{Greedy by Covering Heuristic}\label{alg:covering}
	
	\begin{algorithmic}[H]
		\WHILE{$|U| > 0$}
		\STATE select $s_{i,j} \in S$ which covers the the greatest number of $u \in U$
		\STATE add all covered $u$ to $K$
		\STATE remove all covered $u$ from $U$
		\STATE remove $s_{i,j}$ from $S$
		\ENDWHILE
		
	\end{algorithmic}
\end{algorithm}

\section{Computational Results}\label{computational}

In this section, we evaluate the performance of the product $3$-SAT model and covering heuristic formulated to apply the global variable substitution method to large $3$-SAT instances. Since the circuit depth for one iteration of QAOA is the maximum degree plus one or the maximum degree plus two, in all cases, we take the circuit depth to be the maximum degree plus two since is the upper bound of the circuit depth for one iteration. Thus, we compare the upper bound of the circuit depth for one iteration achieved using the integer programming model and covering heuristic, which we denote $\Delta_{I.P}'$ and $\Delta_{C}'$ respectively, to the upper bound of the circuit depth, denoted $\Delta_{L}'$, calculated using the linear model described in Section \ref{linear}. We apply each formulation and method to $3$-SAT problem instances from the SATLIB Benchmark Suite developed by Holger Hoos and Thomas Stütze \cite{hoos2000satlib}. This suite is composed of thousands of SAT instances of varying families and sizes.  We choose the first instance of each problem set to evaluate.

\begin{table}
    \centering
    \begin{tabular}{|c|c|c|c|c|c|}
    \hline
    \textbf{Problem Name} & \textbf{Linear Formulation} & \textbf{IP Solution}       & \textbf{Covering Heuristic} \\
              & $\Delta_{L}'$ & ($\Delta_{IP}'$, \#Subs) & ($\Delta_{C}'$, \#Subs)\\ \hline
uf20-91 & 84 & (34, 39) & (42, 45)\\ \hline
uf50-218 & 100 & (41, 141) & (54, 143) \\ \hline
uf75-325 & 115 & (48, 208) & (60, 214) \\ \hline
uf100-430 & 103 & (46, 330) & (60, 329) \\ \hline
uf125-538 & 133 & (56, 417) & (71, 419) \\ \hline
uuf50-218 & 95 & (41, 135) & (50, 141)  \\ \hline
uuf75-325  & 106 & (43, 233) & (60, 240)  \\ \hline
uuf100-430 & 105 & (47, 337) & (65, 334) \\ \hline
uuf125-538 & 126 & (54, 429) & (73, 430) \\ \hline
RTI\_k3\_n100\_m429 & 132 & (61, 332) & (75, 334) \\ \hline
BMS\_k3\_n100\_m289 & 110 & (53, 240) & (63, 236) \\ \hline
CBS\_k3\_n100\_m403\_b10  & 94 & (43, 315) & (56, 309)\\ \hline
CBS\_k3\_n100\_m403\_b30 & 94 & (45, 318) & (60, 309)\\ \hline
CBS\_k3\_n100\_m403\_b50 & 96 & (43, 317) & (61, 319)\\ \hline
\hline
\end{tabular}
    \caption{The circuit depth upper bound for each SATLIB $3$-SAT instance using the linear formulation and product formulation with GVS methods. The IP and Covering Heuristic columns also displays the number of substitutions made to cover every $c \in C$. }
    \label{tab:Uniform3SAT}
\end{table}

The results in Tab.~\ref{tab:Uniform3SAT} indicate that the graphs derived using the product formulation combined with the GVS method result in significantly lower circuit depth for one iteration of QAOA than the graphs derived from the linear formulation. For every problem instance, $\Delta_{IP}' < 1/2\Delta_{L}$. The heuristic does not yield as large of a reduction, but most instances see a significant reduction in degree compared to the linear method. It is interesting to note that as the size of each $3$-SAT instance increases, the $\Delta$ of the derived graph does not increase significantly. We can attribute this to the uniformity in the ratio of the literals to clauses and the uniformity of the distribution of those literals amongst all problem instances.

\begin{figure} 
	\centering
	\begin{subfigure}{.45\textwidth}
  \centering
  \includegraphics[width=1\linewidth]{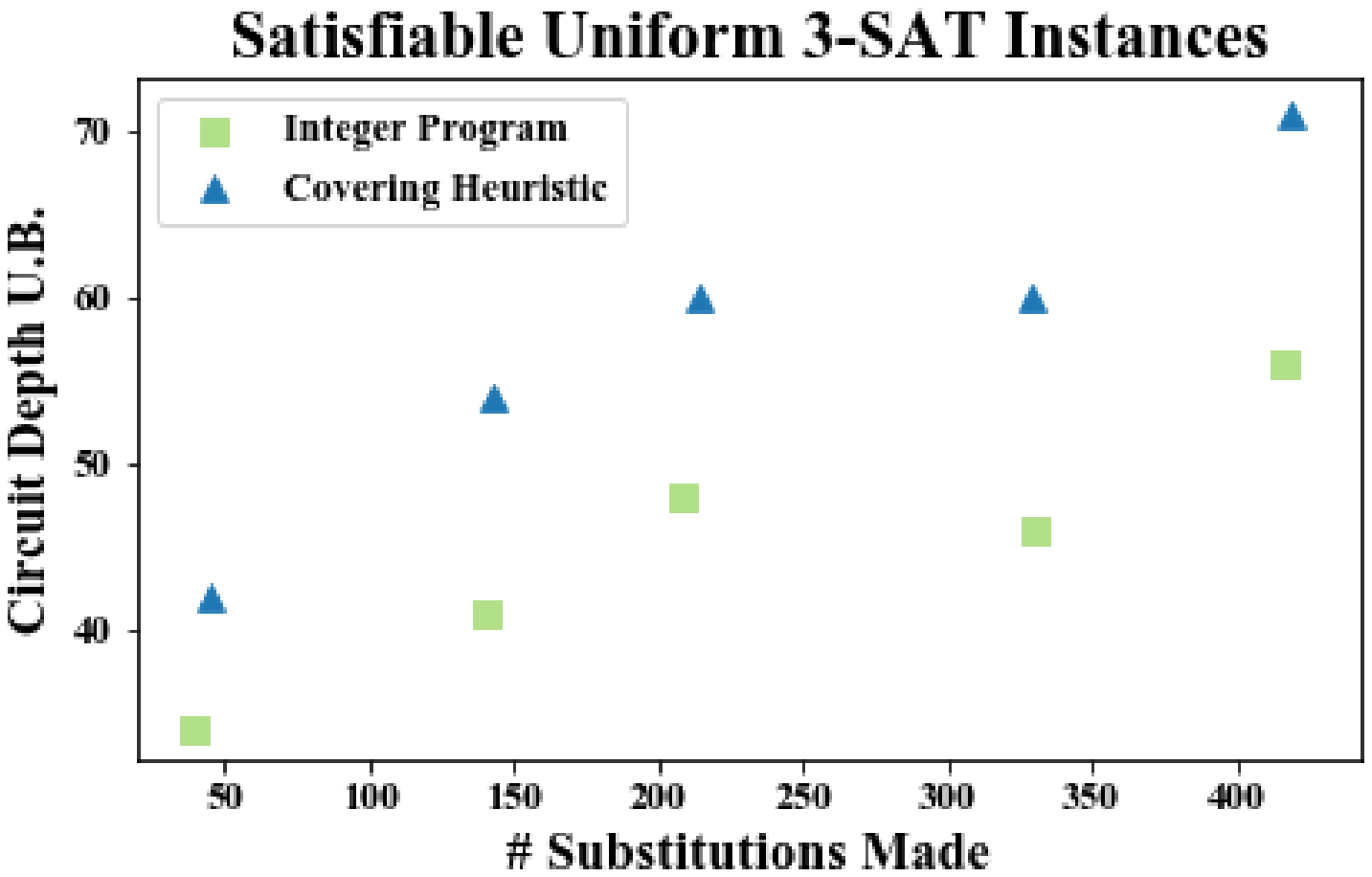}
   \caption{Plot of the maximum degree of the derived graph versus the number of substitutions made for the satisfiable uniform $3$-SAT instances in Tab.~\ref{tab:Uniform3SAT}.}
   \label{fig:DegVSub_SAT}
\end{subfigure}
\begin{subfigure}{.45\textwidth}
  \centering
  \includegraphics[width=1\linewidth]{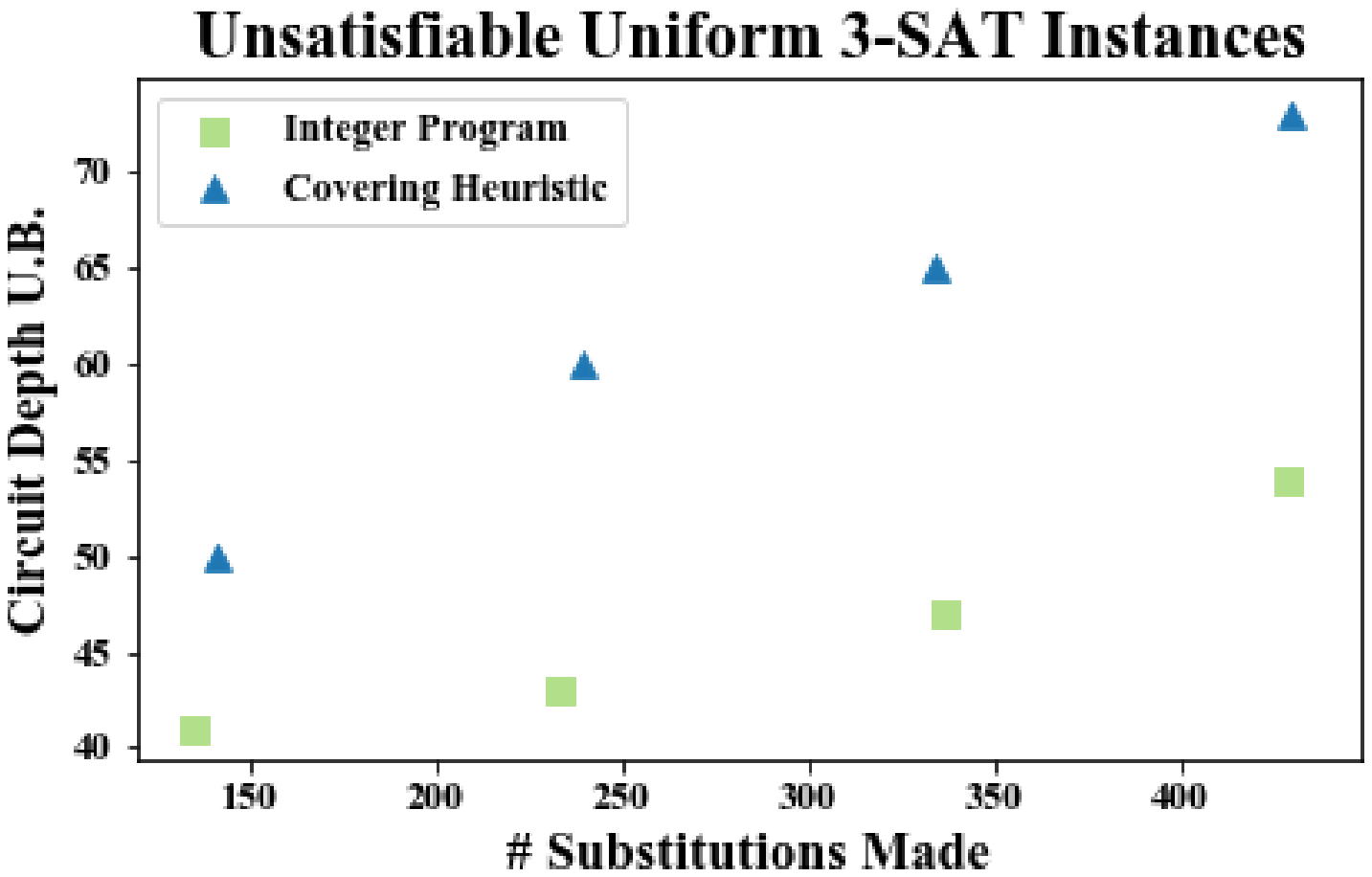}
   \caption{Plot of the maximum degree of the derived graph versus the number of substitutions made for the unsatisfiable uniform $3$-SAT instances found in Tab.~\ref{tab:Uniform3SAT}.}
   \label{fig:DegVSub_UNSAT}
\end{subfigure}

\caption{Plots for uniform $3$-SAT instances found in Tab.~\ref{tab:Uniform3SAT}. The satisfiable instances plotted are denoted 'uf$n$-$|C|$'. The unsatisfiable instances are denoted 'uuf$n$-$|C|$'. These instance names indicate the number of literals $n$ and clauses $C$ in each problem instance. We choose the maximum degree to be the y-axis since the circuit depth for one iteration is either the maximum degree plus one or the maximum degree plus two.}
\end{figure} 

In nearly every $3$-SAT instance, the GVS integer program makes more than or approximately equal to the amount of substitutions made by the covering heuristic method. However, the maximum degree of the IP derived graphs $\Delta_{IP}$ are significantly lower than $\Delta_{C}'$. This trend can be seen in Figs.~\ref{fig:DegVSub_SAT}-~\ref{fig:DegVSub_UNSAT} which plot the maximum degree of the derived graph against the number of substitutions made for the uniform $3$-SAT instances. This seems to indicate that simply minimizing the number of substitutions does not necessarily minimize the maximum degree of the GVS derived graph. In particular the uf50-218 instance from \cite{hoos2000satlib}, which is a satisfiable instance with $50$ variables and $218$ clauses, achieves $\Delta_{IP}' = 41$, which is accomplished by making $141$ substitutions. The covering heuristic makes three more substitutions than the LP, but produces $\Delta_{C}' = 54$. The unsatisfiable instance of the same size, uuf50-218, achieves $\Delta_{IP}' = 41$ by making 138 substitutions. The covering heuristic in this instance made only six more substitutions, but resulted in a  $\Delta_{C}' = 50$. To investigate this trend, we plot the distribution of the degrees of each vertex in the derived graph for each method.

	\begin{figure}
	\centering
	\begin{subfigure}{.45\textwidth}
  \centering
  \includegraphics[width=1\linewidth]{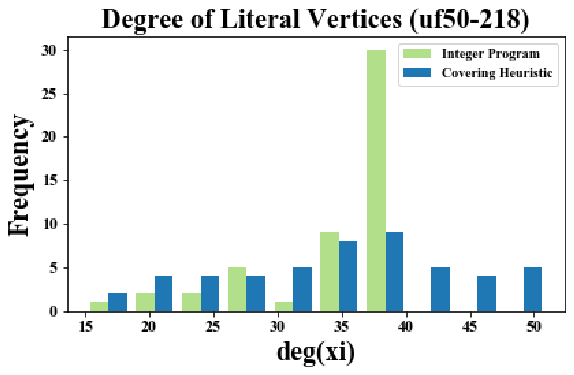}
   \caption{Histogram of the degrees of each literal vertex $x_{i}$ in the graph derived from SAT instance uf50-218}
   \label{fig:litDist_SAT}
\end{subfigure}
\begin{subfigure}{.45\textwidth}
  \centering
  \includegraphics[width=1\linewidth]{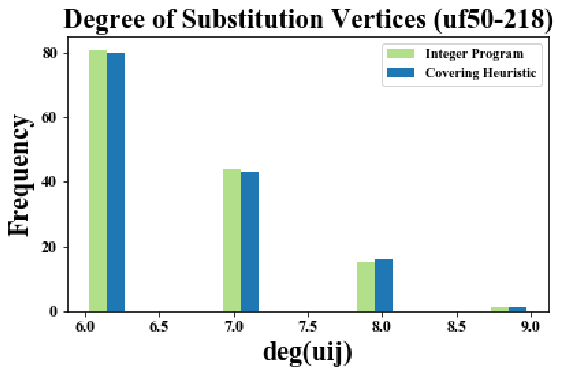}
   \caption{Histogram of the degrees of each substitution vertex $u_{i,j}$ in the graph derived from SAT instance uf50-218}
   \label{fig:subDist_SAT}
\end{subfigure}
\begin{subfigure}{.45\textwidth}
  \centering
  \includegraphics[width=1\linewidth]{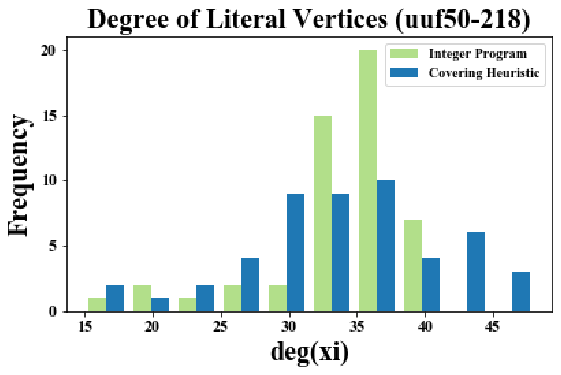}
   \caption{Histogram of the degrees of each substitution vertex $u_{i,j}$ in the graph derived from SAT instance uuf50-218}
   \label{fig:litDist_UNSAT}
\end{subfigure}
\begin{subfigure}{.45\textwidth}
  \centering
  \includegraphics[width=1\linewidth]{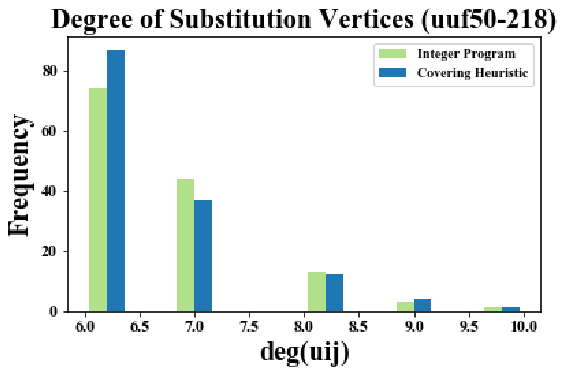}
   \caption{Histogram of the degrees of each substitution vertex $u_{i,j}$ in the graph derived from SAT instance uuf50-218}
   \label{fig:subDist_UNSAT}
\end{subfigure}
\caption{Histograms for specific uniform $3$-SAT problems. We plot the degrees of each vertex since the circuit depth for one iteration is either the maximum degree plus one or the maximum degree plus two.}
\end{figure}

As shown in Fig.~\ref{fig:litDist_SAT}- \ref{fig:subDist_UNSAT}, the degrees of the literal vertices $v_{i}$ are larger than those of the substitution vertices $u_{i,j}$ for each $3$-SAT instance we evaluated. We can attribute this to the GVS method of determining the degree for each vertex in the derived graph and the sparseness of the literals in each instance. Each unique pair $s_{i,j}$ which is substituted will add three or four edges to the degree of the literal vertices $v_{i}$ and $v_{j}$. However, making this substitution only adds one edge to the corresponding substitution vertex $u_{i,j}$. In both instances displayed here, a literal $x_{i}$ is included in approximately thirteen clauses. Clearly, this severely limits the amount of clauses each substitution is able to cover. Consequently, this significantly increases the degree of each literal vertex as more unique substitutions are required to cover all clauses and simultaneously limits the degree of each substitution vertex $u_{i,j}$ since each $s_{i,j}$ is used very few times. We can see this represented in Figures \ref{fig:subDist_SAT} and \ref{fig:subDist_UNSAT}  as nearly $40\%$ of the substitutions made in both instances only cover one clause. The most clauses covered by a any substitution is five.

The distribution of literal vertices of the integer program differs significantly from the distribution of the heuristics in both problem instances. A majority of the literal vertices in the graph derived from the IP take on the value of $\Delta_{IP}'$. It is clear that the integer program is not simply minimizing the number of substitutions made, but rather appears to limit the amount of substitutions per literal $x_{i}$. For these sparse and uniform $3$-SAT instances, the best method of minimizing the max degree of any $v_{i}$ is to attempt to distribute the number of substitutions made evenly amongst all literal vertices $v_{i}$.

\section{Discussion}\label{conclusion}

In this paper, we analyze an approach to minimizing the circuit depth of the quantum approximate optimization algorithm by expressing general combinatorial optimization problems in varying forms.  We compare a linear formulation that is the natural choice in conventional optimization algorithms with a product formulation that we posit as natural for QAOA. The product formulation leads to monomials in more than two variables, which cannot be directly implemented on a quantum computer with two qubit gates. Thus, we introduce the global variable substitution method to decompose them into two variable terms which can be implemented on a quantum computer with two qubit gates. For each of these formulations, we analytically compute the circuit depth in terms of the maximum degree of a graph derived from the problem instance and formulation. We demonstrate that the product formulation gives shallower circuits then the linear formulation for benchmark $3$-SAT problems.

 The global variable substitution requires constraints that must be satisfied in order to obtain the optimal solution, as does the linear formulation. We can derive graphs for the linear and product formulations from the objective function and the appropriate constraints. The circuit depth is directly related to the maximum degree of the derived graph. 

We evaluate the circuit depth of the product formulation with global variable substitutions by writing an integer program that computes the minimal circuit depth of the linear formulation and product formulation for a collection of benchmark problems. In all cases, the product formulation gives circuit depth roughly half that of the linear formulation. The linear formulation for $3$-SAT requires exactly three ancillary qubits per clause, where the product formulation requires four per substitution, although substitutions can sometimes be reused to reduce the number of ancillary qubits.  We find several additional interesting features of the approach.

We find that minimizing the number of substitutions per problem instance does not necessarily minimize the maximum degree. For example, when solving ``uf-100-430", the covering heuristic makes $329$ substitutions for a maximum degree of $60$, whereas the IP makes $330$ substitutions for a maximum degree of $46$. We also note that the objective function for the IP can be modified to limit the number of substitutions. While this may drive up the degree of vertices, it also reduces the number of ancillary qubits, as each substitution requires four additional qubits. Thus, the problem formulation can be changed to accommodate different hardware. 

The focus of this work has been on using the product formulation of $3$-SAT instances to minimize QAOA circuit depth relative to a conventional linear formulation. Extending the analysis of linear and product formulations to more general problems will help determine additional types of problems that benefit from this approach. Additionally, there may be other formulations for specific problems that result in shallower circuits than the linear or product formulations. While the product formulation with GVS gives shallower circuits for $3$-SAT, future work should determine if the reformulation gives a comparable outcome to the linear formulation in the same number of QAOA iterations. A final note is that the global variable substitution method can be used to rewrite problems in terms of gates acting on $m$ qubits. If more general gates become available on quantum computers, then a similar analysis could lead to new approaches for minimizing depth.

\acknowledgments

This work was supported by DARPA ONISQ program under award W911NF-20-2-0051. J.\ Ostrowski acknowledges the Air Force Office of Scientific Research award, AF-FA9550-19-1-0147. G.\  Siopsis  acknowledges the Army Research Office award W911NF-19-1-0397. J.\ Ostrowski and G.\ Siopsis acknowledge the National Science Foundation award OMA-1937008.

This manuscript has been authored by UT-Battelle, LLC under Contract No. DE-AC05-00OR22725 with the U.S. Department of Energy. The United States Government retains and the publisher, by accepting the article for publication, acknowledges that the United States Government retains a non-exclusive, paid-up, irrevocable, world-wide license to publish or reproduce the published form of this manuscript, or allow others to do so, for United States Government purposes. The Department of Energy will provide public access to these results of federally sponsored research in accordance with the DOE Public Access Plan. (http://energy.gov/downloads/doe-public-access-plan).

\bibliographystyle{unsrt}
\bibliography{references}
\end{document}